\documentclass[11pt]{article}
\usepackage{ifthen}
\usepackage{amsmath, amssymb}
\usepackage{framed}
\usepackage{amsthm}

\usepackage{ifpdf}

\usepackage{fullpage}
\usepackage[utf8]{inputenc}

\usepackage{algorithm,algpseudocode}
\algtext*{EndWhile} 
\algtext*{EndIf} 
\algtext*{EndFor} 
\usepackage{amsfonts}
\usepackage{mathrsfs}
\usepackage{comment}
\usepackage{mathtools}

\usepackage{enumitem}

\usepackage{xcolor}

\usepackage{mdframed}
\usepackage{tikz}

\usepackage{bbm}
\usepackage{dsfont}
\usepackage{hyperref}
\usepackage{subcaption}

\theoremstyle{plain}
\newtheorem{theorem}{Theorem}[section]
\newtheorem{inequality}[theorem]{Inequality}
\newtheorem{lemma}[theorem]{Lemma}

\newtheorem{observation}[theorem]{Observation}

\theoremstyle{definition}

\newcommand{\otilde}{\tilde{O}}
\newcommand{\col}{\mathbf{col}}
\usepackage{cleveref}
\usepackage{lipsum}

\begin{document}
	
	\title{Dynamic $(\Delta+1)$-Coloring Against an Adaptive Adversary:\\
    Faster, Simpler, Worst-Case, and Bounded Recourse\footnote{To appear in SODA 2027.}
    }
	\author{{Yuval Gil}\footnote{This research was done while the author was a postdoc at Reykjavik University. yuval.gil@gmail.com} \and {Magnús M. Halldórsson}\footnote{Reykjavik University, Iceland. Supported by the Icelandic Research Fund, Grant 2511609. mmh@ru.is}}
	
	\date{}
	
	\maketitle
	
	\begin{abstract}
    We study fully dynamic $(\Delta+1)$-coloring against an adaptive adversary. Behnezhad, Rajaraman, and Wasim (SODA 2025) gave the first algorithm with $\tilde O(n^{8/9})$ amortized update time, later improved to $\tilde O(n^{2/3})$ by Flin and Halldórsson (ICALP 2025). We give a fully dynamic $(\Delta+1)$-coloring algorithm with $\tilde O(\sqrt n)$ \textbf{worst-case} update time. Unlike previous algorithms, it avoids sparse-dense decompositions and relies only on elementary data structures, resulting in a short self-contained algorithm and analysis. Moreover, it avoids global recoloring, resulting in a worst-case recourse of $3$.
\end{abstract}
	\paragraph{Independent and concurrent work.}
	In an independent concurrent work \cite{AssadiY27}, Assadi and Yazdanyar present a fully dynamic $(\Delta+1)$-coloring algorithm with an amortized update time of $\tilde{O}(\sqrt{n})$ against an adaptive adversary. To the best of our knowledge, the two works use entirely disjoint techniques.
	\thispagestyle{empty}
	\clearpage
	\thispagestyle{empty}
	\tableofcontents
	\clearpage
	\setcounter{page}{1}


\section{Introduction}\label{section:introduction}
The problem of $(\Delta+1)$-coloring a graph of maximum degree $\Delta$ is one of the most fundamental problems in the context of graph algorithms. While the problem admits a simple greedy solution in the standard sequential setting, it becomes more challenging in several modern computational models. Consequently, the problem has been studied extensively in restricted computational models
such as distributed \cite{Linial92,BarenboimEPS16,HarrisSS18,CLP20,HalldorssonKMT21}, massively parallel \cite{CFGUZ19,CzumajDP21}, streaming and sublinear time \cite{AssadiCK19,AssadiY26}, and communication complexity \cite{FlinM25,ChangMNS25}. 

We focus on $(\Delta+1)$-coloring in the \emph{fully dynamic} setting in which the input graph undergoes dynamic updates in the form of edge insertions and deletions. The main performance measure of a dynamic algorithm is its \emph{update time} which is defined as the time required by the algorithm to respond to an update. Another important measure is the \emph{recourse}, i.e., the number of output changes per update.

A straightforward implementation of the greedy algorithm has update time $O(\Delta)$, which is $\Theta(n)$ when $\Delta=\Theta(n)$. No deterministic algorithm is known to achieve $o(n)$ update time while using only $\Delta+1$ colors.
For this reason, the research of dynamic $(\Delta+1)$-coloring has mainly focused on randomized algorithms.

The first non-trivial randomized algorithm for $(\Delta+1)$-coloring achieves amortized expected update time $O(\log\Delta)$ \cite{BCHN18}. This guarantee assumes an \emph{oblivious} adversary, which fixes the update sequence before seeing the algorithm's random choices. The update time was subsequently improved to a constant 
in \cite{HP22} (in expectation) and \cite{BGKL22} (in expectation and whp).\footnote{An event $A$ holds \emph{with high probability (whp)} if $\Pr[A]\geq 1-n^{-c}$ for an arbitrarily large constant $c$.}

We instead consider an \emph{adaptive adversary}, which may choose each update after observing the algorithm's previous responses. 
Our guarantees continue to hold even if the adversary observes all random bits generated during previous updates, and therefore all data structures maintained by the algorithm.

This seemingly modest change fundamentally alters the problem. Until very recently, no algorithm was known to beat the trivial $O(n)$ update time.
The first algorithm to achieve an improvement over the trivial bound was presented in \cite{BehnezhadRW25} by Behnezhad, Rajaraman, and Wasim. The algorithm is randomized and achieves an amortized update time of  $\otilde(n^{8/9})$. Shortly after, Flin and Halld{\'{o}}rsson \cite{FlinH25} improved the amortized update time to $\otilde(n^{2/3})$. 

From a technical standpoint, the works of \cite{BehnezhadRW25} and \cite{FlinH25} both rely on the algorithmic framework of sparse-dense decompositions \cite{Reed98}. In \cite{FlinH25}, it is observed that an update time of $O(n^{2/3})$ poses a natural barrier for this approach. This led the authors to ask whether an algorithm with an update time of $o(n^{2/3})$ exists \cite[Problem 1]{FlinH25}.

\paragraph{Our contributions.}
We answer this question affirmatively.
More concretely, we obtain the following:
\begin{theorem}\label{theorem:worst-case}
	There exists a fully dynamic algorithm that maintains a $(\Delta+1)$-coloring against an adaptive adversary in worst-case update time $\otilde(\sqrt{n})$ whp. Furthermore, the algorithm's worst-case recourse is $3$.
\end{theorem} 
The improved update time is accompanied by several additional advantages. Unlike previous algorithms, it guarantees the update time in the worst case and has worst-case recourse 3, as it avoids periodic global recoloring. Equally importantly, it dispenses with sparse--dense decompositions entirely, relying only on elementary data structures and a short self-contained analysis.

Beyond its conceptual appeal, this simplification is also significant from a practical perspective. Sparse--dense decompositions rely on sophisticated machinery and hide substantial constant factors, whereas our algorithm uses only elementary data structures. Since the adaptive-adversary model captures many realistic dynamic scenarios \cite{CherapanamjeriS23}, we believe this simplicity makes the algorithm more attractive both theoretically and practically.

Our result also identifies what appears to be a natural limit of the random color-trial paradigm that underlies essentially all known randomized $(\Delta+1)$-coloring algorithms.

\paragraph{Can we do better than $O(\sqrt{n})$?} 
As observed in \cite{FlinH25}, $O(\sqrt{n})$ is an inherent limit for the update time of \emph{natural} $(\Delta+1)$-coloring algorithms against an adaptive adversary. The reason is that $\sqrt{n}$ is the meeting point of two natural strategies for \emph{testing whether a given color $i$ is available} for a given node $v$. The first strategy takes $O(\Delta)$ time and involves going over all of $v$'s neighbors and checking if any of them is colored with $i$; the second strategy takes $O(n/\Delta)$ and involves going over all nodes colored $i$ and checking if any of them is $v$'s neighbor. 

Random color trials have become the central primitive for randomized $(\Delta+1)$-coloring algorithms. They enable efficient discovery of available colors without explicitly scanning neighborhoods or color classes, and underlie algorithms in the distributed setting \cite{BarenboimEPS16,SW10,HarrisSS18,CLP20,HalldorssonKMT21,HKNT22}, the streaming model \cite{AssadiCK19}, and the fully dynamic setting against an oblivious adversary \cite{BCHN18,HP22,BGKL22}.
This observation suggests that the $\sqrt n$ barrier is not merely an artifact of the current analysis, but of the random color-trial paradigm itself. 
Together, these observations suggest that breaking the $\sqrt n$ barrier will require techniques that circumvent the random color-trial paradigm.

An alternative --- and perhaps more realistic --- direction is to establish lower bounds for the problem. In this regard, a recent breakthrough of \cite{BBFKS26} proves conditional lower bounds for dynamic maximal independent set and maximal clique against adaptive adversaries. The authors identify proving analogous lower bounds for dynamic $(\Delta+1)$-coloring as a particularly compelling open problem.
A first step would be to give a non-trivial lower bound for dynamic \emph{deterministic} $(\Delta+1)$-coloring algorithms.

\paragraph{Technical overview.} We start by describing a very simple \emph{amortized} algorithm. Throughout this overview, let us assume that $\Delta=\Omega(\sqrt{n}\log n)$ and that the degree of all nodes is $\Theta(\Delta)$. Here, notice that nodes of degree $o(\Delta)$ are anyway easy to recolor as they have $\Theta(\Delta)$ available colors. 

The algorithm operates in phases of $n$ updates. In every phase, the algorithm maintains a set $C\subseteq V$ of \emph{chosen} nodes. Let us first describe the algorithm under the (incorrect) assumption that throughout the phase, every node $v\in V$ has $\Theta(\sqrt{n}\log n)$ chosen neighbors. Notice that this allows the algorithm to maintain \emph{full information} about the colors in the neighborhood of every  chosen node $v\in C$. This is because whenever a node $v\in V$ changes its color, the algorithm can update the information of all its chosen neighbors within the desired time bound. Thus, recoloring chosen nodes is \emph{trivial}.

As for recoloring unchosen nodes, the idea is simply to \emph{steal} a color from a chosen neighbor, and then recolor that neighbor (which is easier since it is chosen). That is, when an unchosen node $v\in V-C$ is being recolored, the algorithm looks for a color to steal by trying colors uniformly at random until finding a color $i\in [\Delta+1]$ that satisfies: (1) at most $1$ chosen neighbor of $v$ is colored $i$; and (2) no unchosen neighbor of $v$ is colored $i$. To be able to check (2), the algorithm maintains a list of all unchosen nodes colored $i$ for each color $i\in [\Delta+1]$. Furthermore, the algorithm explicitly avoids the \emph{expensive} colors, i.e., colors for which the list length exceeds $10\sqrt{n}$. 

Since the number of expensive colors is bounded by $\sqrt{n}/10$, and the number of chosen neighbors is $\Omega(\sqrt{n}\log n)$, it follows from a simple combinatorial observation that the algorithm finds a color to steal after $O(\Delta\log n/\sqrt{n}\log n)=O(\Delta/\sqrt{n})$ trials whp. Here, at a first glance, it may seem like the recoloring does not gain any advantage over the trivial $O(\Delta)$ procedure: there are $O(\Delta/\sqrt{n})$ many trials, and on each of them the procedure may spend $\Theta(\sqrt{n})$ time. However, we provide a refined analysis based on Bernstein's inequality asserting that this is not the case --- the total time spent by the recoloring procedure is $O(\sqrt{n}\log n)$ whp. 

As alluded to above, the assumption that every node $v\in V$ has $\Theta(\sqrt{n}\log n)$ chosen neighbors is too optimistic. Note, however, that in the static setting this assumption can be easily satisfied whp via a procedure that samples each node into the chosen set $C$ with probability $\Theta(\sqrt{n}\log n/\Delta)$. Therefore, the algorithm can start each phase with a preprocessing procedure that computes a chosen set $C$ with the desired properties, and charge its cost (which is $\otilde(n^{3/2})$ including updates to the relevant data structures) to the rest of the phase. 

This leaves open the challenge of (approximately) keeping the assumption alive through the phase. To that end, whenever an unchosen node violates the assumption, it is simply added to the chosen set $C$. As we show, the number of such violations is $O(\sqrt{n})$ which allows to uphold the amortized update time for these additions into $C$. This leaves chosen nodes as the only possible violators of the assumption. To handle that, we observe that the algorithm can be modified to work even if the \emph{arboricity} of the induced subgraph $G[C]$ is $O(\sqrt{n}\log n)$ (instead of its maximum degree), which is actually the case. 

The \emph{worst-case} algorithm is obtained by deamortizing the algorithm above. As a first step, we avoid any additions into the chosen set $C$ during the phase. To that end, we classify the additions into $C$ in the amortized algorithm to two node-sets: the \emph{underloaded} nodes that have too few chosen neighbors, and the \emph{overloaded} nodes that have too many chosen neighbors. 

The issue with underloaded nodes is that when recoloring them, the success probability of the random color trials approach described above becomes too small. To overcome this obstacle, we define a designated recoloring procedure for these nodes that simply aims to steal a color from a non-underloaded neighbor (\`a la the color stealing mechanism implemented for unchosen nodes), and then recolor this neighbor (which is easier since it is not underloaded). 

For overloaded nodes, the issue is that when they change color, the algorithm cannot update the relevant data structures of their chosen neighbors (as there are too many of them). Instead, chosen nodes do not store the information on their overloaded neighbors --- they only collect it when being recolored. For both the underloaded and overloaded nodes, the fact that there are only $O(\sqrt{n})$ of them keeps the update time of all recolorings at the desired $\otilde(\sqrt{n})$. 

The final deamortization step is in disposing of the costly preprocessing stage in the beginning of each phase. The idea here is simple --- the preprocessing of each phase is implemented across the previous phase. Intuitively, since the preprocessing time is $\otilde(n^{3/2})$, its implementation can be distributed across the $n$ updates of a phase such that the $\otilde(\sqrt{n})$ update time is preserved. Moreover, this preprocessing module can also handle edge insertion/deletion by updating its relevant data structures accordingly. This is done in an analogous way to how the main module updates its data structures, thus maintaining the worst-case update time and correctness.

\subsection{Organization}\label{section:paper-organization}
Section \ref{section:preliminaries} introduces the notation. Section \ref{section:amortized-alg} presents a simple amortized algorithm capturing the main ideas, and Section \ref{section:worst-case-alg} deamortizes it to obtain Theorem \ref{theorem:worst-case}.

\paragraph{Remark.} The algorithms of Sections \ref{section:amortized-alg} and \ref{section:worst-case-alg} are presented under the assumption that $\Delta$ is known to the algorithm in advance. We emphasize that this assumption is made strictly for the sake of a clear exposition. The adjustments required to lift this assumption are presented in Section \ref{section:knowledge}.

    
\section{Preliminaries}\label{section:preliminaries}
\paragraph{Notations and definitions.}
Consider a graph $G=(V,E)$ with maximum degree $\Delta$. A function $\col:V\to[\Delta+1]$ is said to be a \emph{proper} $(\Delta+1)$-coloring if $\col(u)\neq \col(v)$ for every edge $(u,v)\in E$. Given a proper $(\Delta+1)$-coloring $\col$ and a subset $S\subseteq V$ of nodes, let $\col(S)=\{\col(v)\mid v\in S\}$ be the set of colors assigned by $\col$ to the nodes of $S$. In this context, the \emph{color histogram} of $S$ is a set of $\Delta+1$ pairs $(i,\mu_{i}(S))$ where $i\in [\Delta+1]$ is a color and $\mu_{i}(S)=|\{v\in S\mid \col(v)=i\}|$ is its number of occurrences in $S$.

As common, for a node $v\in V$ we denote by $N(v)$ the set of $v$'s neighbors and let $\deg(v)=|N(v)|$. The following notions are defined w.r.t.\ a proper $(\Delta+1)$-coloring $\col$. A color $i\in [\Delta+1]$ is said to be \emph{available} for a node $v\in V$ if $i\notin \col(N(v))$. For a subset $Q\subseteq N(v)$ of $v$'s neighbors, we say that a color $i\in [\Delta+1]$ is \emph{$Q$-good} for $v$ if $i\notin \col(N(v)-Q)$ (i.e., $\mu_i(N(v)-Q)=0$) and $i$ is assigned to at most one node in $Q$ (i.e., $\mu_i(Q) \leq 1$). Observe that by definition, if a color is available for $v$, then it is also $Q$-good for $v$ for any subset $Q\subseteq N(v)$. We note that the notion of $Q$-good colors is similar to that of \emph{blank/unique} colors from previous works on dynamic algorithms against an \emph{oblivious adversary} \cite{BCHN18,BGKL22,HP22}. In particular, those works make use of the following structural observation (which we prove for completeness). This observation will become useful for our algorithm as well.
\begin{observation}\label{observation:q-good}
    Let $\col:V\to[\Delta+1]$ be a proper $(\Delta+1)$-coloring of a graph $G=(V,E)$. Let $v\in V$ be an arbitrary node and let $Q\subseteq N(v)$ be some subset of $v$'s neighbors of size at least $2k$ for some integer $k\geq 1$. Then, at least $k+1$ colors are $Q$-good for $v$.
\end{observation}
\begin{proof}
    Let $d=|\col(Q)|$. If $d\leq k$, then the total number of distinct colors given to $v$'s neighbors is at most $\Delta-|Q|+d\leq \Delta-2k+k=\Delta-k$ which means that there are at least $k+1$ colors that are available for $v$ (and hence, also $Q$-good for $v$). So, we are left to consider the case $d>k$. Let $Q'\subseteq Q$ be a set consisting of $k$ nodes in $Q$ such that $|\col(Q')|=k$, i.e., the nodes of set $Q'$ are assigned $k$ distinct colors. Let us partition $Q'$ into the set $Q'_{1}$ of nodes whose color is not assigned to any node in $N(v)-Q'$, and the set $Q'_{2}=Q'-Q'_{1}$. Note that by definition, the colors of nodes in $Q'_{1}$ are $Q$-good for $v$ and are not available for $v$. Moreover, since each color given to a node $u\in Q'_{2}$ is assigned also to a node in $N(v)-Q'$, the number of distinct colors assigned to $v$'s neighbors is at most $\Delta-|Q'_{2}|$. Hence, there are at least $|Q'_{2}|+1$ available colors for $v$. In total, the number of $Q$-good colors for $v$ is at least $|Q'_{1}|+|Q'_{2}|+1=|Q'|+1=k+1$.
\end{proof}

\paragraph{Concentration inequalities.}
We make use of the following inequality.

\begin{inequality}[Bernstein's inequality {{\cite[Theorem~2.8.4]{Vershynin2018}}}]\label{inequality:bernstein-standard}
Let $X_{1},\dots, X_{n}$ be independent zero-mean random variables and let $X=\sum_{i=1}^{n}X_{i}$. Suppose that $X_{i}\leq M$ for all $1\leq i\leq n$. Then, for all positive $t$,
\[\Pr[X\geq t]\leq \exp\left(-\frac{t^{2}/2}{\sum_{i=1}^{n}\mathbb{E}[X_{i}^{2}]+Mt/3} \right)\]
\end{inequality}

Notice that given a (not necessarily zero-mean) random variable $Y$, one can define a zero-mean random variable $X=Y-\mathbb{E}[Y]$. Based on this transformation, we get the following special case of Bernstein's inequality which would be more convenient for our application.
\begin{inequality}\label{inequality:bernstein-convenient}
Let $X_{1},\dots, X_{n}$ be independent random variables and let $X=\sum_{i=1}^{n}X_{i}$. Suppose that $X_{i}-\mathbb{E}[X_{i}]\leq M$ for all $1\leq i\leq n$. Then, for all positive $t$,
\[\Pr[X-\mathbb{E}[X]\geq t]\leq \exp\left(-\frac{t^{2}/2}{\sum_{i=1}^{n}\textup{Var}(X_{i})+Mt/3} \right)\]
\end{inequality}
	
\section{Simple Algorithm with $\otilde(\sqrt{n})$ Amortized Update Time}\label{section:amortized-alg}
In this section, we give a simple algorithm that achieves an amortized bound on the update time. That is, we constructively prove the following theorem. 
\begin{theorem}\label{theorem:amortized-alg-statement}
	There exists a fully dynamic algorithm that maintains a $(\Delta+1)$-coloring against an adaptive adversary in amortized update time $\otilde(\sqrt{n})$ whp. Furthermore, the algorithm's worst-case recourse is $2$.
\end{theorem}
\noindent This algorithm constitutes the first step towards our main result --- an algorithm with a worst-case update time of $\otilde(\sqrt{n})$ and a recourse of $3$ --- presented in Section \ref{section:worst-case-alg}.

Similarly to the previous $(\Delta+1)$-coloring algorithms in the adaptive adversary setting \cite{BehnezhadRW25,FlinH25}, the algorithm operates in phases. Each phase begins with a preprocessing stage. Unlike \cite{BehnezhadRW25,FlinH25}, the preprocessing stage does not include a global recoloring step. Instead, the coloring from the previous phase is kept. The phase runs for $n$ dynamic updates. 

We assume throughout that $\Delta\geq 20\sqrt{n}\log n$ and set the parameter $s=\sqrt{n}\log n$.\footnote{In the case where $\Delta<20\sqrt{n}\log n$, we can simply resort to the trivial $O(\Delta)$-time algorithm that upon recoloring a node simply checks the colors of all its neighbors and selects an unused color.} As we will show, the time required for preprocessing as well as the total time required to maintain a valid $(\Delta+1)$-coloring during a phase is $\otilde(ns)$ whp, implying an amortized runtime of $\otilde(s)=\otilde(\sqrt{n})$ per update. For the reader's convenience, a glossary of the main data structures used during the execution of the algorithm is given in Table \ref{table:data-structures}.

\subsection{Preprocessing}\label{section:preprocessing}
In the preprocessing stage, the algorithm computes a set $C\subseteq V$ of \emph{chosen} nodes by sampling each node independently with probability $10s/\Delta$. Naturally, we say that a node $v\in V$ is \emph{unchosen} if $v\notin C$. For every unchosen node $v\in V-C$, the algorithm maintains the set $C_{v}=C\cap N(v)$ of chosen neighbors. For each color $i\in [\Delta+1]$, the algorithm builds a list $L_{i}$ of all \emph{unchosen} nodes that were assigned the color $i$. We say that a color $i$ is \emph{expensive} if $|L_{i}|\geq10\sqrt{n}$, and \emph{cheap} otherwise. Cheap colors are stored in a balanced tree data structure $T_{\text{chp}}$. Notice that a uniformly random color can be sampled from $T_{\text{chp}}$ in $O(\log \Delta)$ time.

The algorithm also maintains an orientation of the edges in the subgraph $G[C]$ throughout the phase so that the set $N^{+}(v)$ of $v$'s outgoing neighbors is stored for each chosen node $v\in C$.\footnote{For a node subset $V'\subseteq V$, we denote by $G[V']$ the subgraph of $G$ induced by $V'$.} In the preprocessing stage, the orientation is set arbitrarily. Let $N^{*}(v)= N(v)-N^{+}(v)$ for each chosen node $v\in C$. To complete the preprocessing stage, for each chosen node $v\in C$, the algorithm stores a color histogram of the colors assigned to the nodes in $N^{*}(v)$. This color histogram is stored in a minimum heap $H_{v}$. In the heap $H_{v}$, each element is associated with a color $i\in [\Delta+1]$ and its key is the number of nodes in $N^{*}(v)$ colored with $i$ (i.e., the minimum is taken w.r.t.\ the number of occurrences $\mu_{i}(N^{*}(v))$ of a color $i$ in $N^{*}(v)$). Notice that for each $v\in C$, $H_{v}$ is of size $\Delta+1$ and thus, supports operations such as increasing/decreasing a key in time $O(\log \Delta)$.\footnote{In the amortized case one can use, e.g., a Fibonacci heap to implement the data structures $H_{v}$ such that heap operations will be done in (amortized) $O(1)$ time. This will not matter much in the grand scheme as the algorithm's runtime anyway incurs $\text{polylog}$ factors in other modules. In addition, Fibonacci heaps are not suited for the algorithm with worst-case guarantee which is the main contribution of the paper.} 

\begin{lemma}\label{lemma:preprocessing}
	The runtime of the preprocessing stage is $\otilde(ns)$ whp. Furthermore, the following properties hold whp by the end of the stage: (1) each node has fewer than $12s$ chosen neighbors; (2) each node whose degree is larger than $\Delta-2s$ has more than $3s$ chosen neighbors.     
\end{lemma}     
\begin{proof}
	We start by showing that the properties hold. For Property (1), notice that the expected number of chosen neighbors for each node is at most $10s$. Therefore, by a Chernoff bound argument along with a union bound over the nodes, we get that the property holds with probability $1-\exp(-\Omega(s))$. For Property (2), observe that if a node has degree larger than $\Delta-2s$, then its expected number of chosen  neighbors is at least $(\Delta-2s)\cdot 10s/\Delta>(\Delta/2)\cdot 10s/\Delta=5s$. The argument now follows by a concentration bound similarly to Property (1).
	
	Regarding the runtime, the sampling of $C$ is done in $O(n)$ time. Notice that by Property (1), the maximum degree in $G[C]$ is $O(s)$ whp, which implies that the orientation of $G[C]$ can be computed and stored in time $O(ns)$. Computing the lists $L_{i}$ takes $O(n)$ time. The algorithm maintains the size of each $L_{i}$ which allows it to update $T_{\text{chp}}$ accordingly. Finally, building each heap $H_{v}$ of a chosen node $v\in C$ takes $O(\Delta)$. Since the size of $C$ is $O(ns/\Delta)$ whp, the total runtime of building the heaps $H_{v}$ is $O(ns)$. Building the sets $C_{v}$ of the unchosen nodes $v\in V-C$ also takes $O(ns)$. Overall, the runtime of the preprocessing stage is $\otilde(ns)$ whp.
\end{proof}
\begin{table}[]
	\centering
	\begin{tabular}{ |l|l|l| } 
		\hline
		\textbf{Notation} & \textbf{Data structure} & \textbf{Role} \\ 
		\hline
		$C$ & set & chosen nodes \\ 
		$C_{v}$& list & $C\cap N(v)$ for $v\in V-C$ \\
		$L_{i}$ & list & unchosen nodes colored $i$\\ 
		$T_{\text{chp}}$ & balanced tree & maintains cheap colors\\
		$N^{+}(v)$ & list & outgoing $G[C]$-neighbors of $v\in C$\\
		$N^{*}(v)$ & list & $N(v)-N^{+}(v)$\\
		$H_{v}$ & min-heap & maintains color histogram of $N^{*}(v)$ for $v\in C$\\
		\hline
	\end{tabular}
	\caption{Data structures used by the algorithm.}
	\label{table:data-structures}
\end{table}

In the upcoming subsection, we present the dynamic algorithm employed during the phase. Throughout, we condition on the success of Lemma \ref{lemma:preprocessing}. That is, we assume that the two properties listed in Lemma \ref{lemma:preprocessing} are satisfied after the preprocessing stage.

\subsection{Handling Dynamic Updates During a Phase}\label{section:phase} 
We describe the procedures required in order to maintain a $(\Delta+1)$-coloring under edge insertions/deletions. Apart from the recoloring procedures (Section \ref{section:recoloring}), we also define some structural procedures (Section \ref{section:structural}) that facilitate the desired update times. Particularly, during the phase the algorithm may select some unchosen nodes to be chosen, i.e., add them into the set $C$ (we never remove nodes from $C$ during a phase). Additionally, the algorithm maintains the edge orientation of the induced subgraph $G[C]$ such that throughout the phase $|N^{+}(v)|=O(s)$, and the minimum heap $H_{v}$ stores the color histogram of $N^{*}(v)=N(v)-N^{+}(v)$ for all $v\in C$.
\subsubsection{Recoloring Procedures}\label{section:recoloring}
The algorithm invokes a recoloring procedure only when an edge $(u,v)$ between two nodes of the same color is inserted (for any other update, no recoloring is performed). If one of the endpoints is a chosen node, it is recolored based on the procedure for recoloring chosen nodes given below (if both endpoints are chosen, one of them is selected arbitrarily for recoloring). Otherwise (both endpoints are unchosen), one endpoint is selected arbitrarily and a recoloring is performed based on the procedure for unchosen nodes. Let us denote by $\col$ the coloring right before the recoloring procedure is invoked. 
\paragraph{Recoloring chosen nodes.} Let $v\in C$ be a chosen node that undergoes a recoloring. The algorithm first updates the heap $H_{v}$ to store the color histogram of $N(v)$. Recall that before the recoloring procedure is invoked, $H_{v}$ stores the color histogram of $N^{*}(v)$. Thus, to store the color histogram of $N(v)$, the algorithm simply goes over each node $u\in N^{+}(v)$ and increases the key of $\col(u)$ in $H_{v}$ by $1$. Then, the algorithm sets $v$'s color to be the color at the root of the heap $H_{v}$ (i.e., a color with minimum key in $H_{v}$). Following that, $H_{v}$ is modified back to storing only the color histogram of $N^{*}(v)$. This is done by going over each node $u\in N^{+}(v)$ and decreasing the key of $\col(u)$ in $H_{v}$ by $1$. Finally, for each $u\in N^{+}(v)$, the algorithm updates $H_{u}$ by decreasing the key of $v$'s old color by $1$ and increasing the key of $v$'s new color by $1$.

\paragraph{Recoloring unchosen nodes.} Let $v\in V-C$ be an unchosen node that undergoes a recoloring. Let $i=\col(v)$ be $v$'s current color and recall that $C_{v}=C\cap N(v)$ is the set of $v$'s chosen neighbors. In the beginning of the procedure, the color histogram of $C_{v}$ is stored in a designated dictionary from which the number of times a color $i$ occurs in $C_{v}$ can be queried. Then, to recolor $v$, the algorithm samples cheap colors uniformly at random from $T_{\text{chp}}$ until finding a cheap $C_{v}$-good color for $v$. That is, the algorithm samples cheap colors uniformly until obtaining a color $j$ satisfying: (1) at most one chosen neighbor has color $j$; and (2) no unchosen neighbor belongs to $L_{j}$. Once such a color is found, it becomes $v$'s new color. 


When $v$ changes its color from $i$ to $j$, it is removed from $L_{i}$ and added to $L_{j}$. If necessary, $T_{\text{chp}}$ is modified accordingly. In addition, the heaps $H_{u}$ of $v$'s chosen neighbors $u\in C_{v}$ are updated by decreasing the key of $i$ and increasing the key of $j$ by $1$. Finally, if $v$ has a chosen neighbor $u\in C_{v}$ which is also colored $j$ (by design, there can be at most one such neighbor), then $u$ is recolored based on the procedure for recoloring chosen nodes presented above.  

\paragraph{Analysis.} We now analyze the recoloring procedures. This is done under the assumption that the following invariants hold throughout the phase. Recall that $s=\sqrt{n} \log n$.

\begin{enumerate}
	\begin{framed}
		\item[(1)] Every unchosen node of degree larger than $\Delta-s$ has at least $2s$ chosen neighbors.
		\item [(2)] Every unchosen node has at most $14s$ chosen neighbors.
		\item [(3)] $G[C]$ is oriented such that the out-degree of every node $v\in C$ is at most $16s$.
	\end{framed}
\end{enumerate}
\noindent Later, in Section \ref{section:structural}, we present the structural procedures that allow one to maintain these invariants.

We start the recoloring procedures' analysis from the following observation.
\begin{observation}\label{observation:good-colors}
	Let $v\in V-C$ be an unchosen node that undergoes a recoloring procedure. If Invariant (1) holds, then more than $s/2$ cheap colors are $C_{v}$-good for $v$.
\end{observation}
\begin{proof}
	If $\deg(v)\leq \Delta-s$, then there are at least $s+1$ colors that are available for $v$ and thus are also $C_{v}$-good for $v$. Otherwise, by Invariant (1), we have $|C_{v}|\geq 2s$. In this case, it follows from Observation \ref{observation:q-good} that at least $s+1$ colors are $C_{v}$-good for $v$. Since the total size of all $L_{i}$ is at most $n$ and a color $i$ is expensive if $|L_{i}|\geq10\sqrt{n}$, there can be at most $\sqrt{n}/10<s/2$ expensive colors. Therefore, the number of $C_{v}$-good colors for $v$ that are cheap is more than $s/2$.  
\end{proof}

We can now bound the number of random color trials in the procedure for recoloring an unchosen node.
\begin{lemma}\label{lemma:number-of-trials-until-good}
	Let $v\in V-C$ be an unchosen node that undergoes a recoloring procedure. If Invariant (1) holds, then a cheap $C_{v}$-good color for $v$ is found within $O(\Delta\log n/s)$ random color trials whp.
\end{lemma}
\begin{proof}
	It follows from Observation \ref{observation:good-colors} that the probability of selecting a cheap $C_{v}$-good color by a random trial is larger than $s/(2(\Delta+1))$. Thus, the probability of not finding a good color within $2c(\Delta +1)\log n/s$ independent trials is at most $(1-s/2(\Delta+1))^{2c\log n(\Delta +1)/s}\leq e^{-c\log n}=n^{-c}$.
\end{proof}

We can now establish the correctness and runtime of each recoloring procedure.
\begin{lemma}\label{lemma:recoloring-time}
	If all invariants are satisfied, then both recoloring procedures yield a proper coloring and their runtime is $\otilde(\sqrt{n})$ whp.
\end{lemma}
\begin{proof}
	We first address the procedure that recolors a chosen node $v\in C$. By Invariant (3), $v$ has $O(s)$ outgoing neighbors. Each heap update costs $O(\log \Delta)$, yielding $O(s\log \Delta)=\otilde(\sqrt{n})$ time. For correctness, first observe that throughout the phase $H_{v}$ stores the color histogram of $N^{*}(v)$. When the recoloring procedure is invoked, the algorithm collects in $H_{v}$ the color histogram of $N(v)$. Since there are at most $\deg(v)<\Delta+1$ distinct colors assigned to the nodes of $N(v)$, it follows from the pigeonhole principle that the color at the root of $H_{v}$ is available for $v$ (i.e., its key is $0$).

	We now consider the procedure that recolors an unchosen node $v\in V-C$. For correctness, recall that the procedure samples colors until finding a cheap $C_{v}$-good color for $v$. Then, if there is a node $u\in C_{v}$ with that color, it is recolored based on the recoloring procedure for chosen nodes. Thus, the correctness follows from the correctness of the recoloring procedure for chosen nodes.
	
	As for time complexity, let us first bound the runtime of the procedure \emph{after} a cheap $C_{v}$-good color for $v$ was found. The procedure first removes $v$ from a list $L_{i}$ and adds it to a list $L_{j}$ which takes $O(1)$ time. Additionally, the procedure may add (resp., remove) one color into (resp., from) $T_{\text{chp}}$ in $O(\log \Delta)$ time. Then, for every $u\in C_{v}$, the procedure performs $O(1)$ decrease/increase key operations to $H_{u}$. By Invariant (2), this takes a total of $O(s\log \Delta)=\otilde(\sqrt{n})$ time. Finally, the algorithm may perform a single recoloring of a chosen node $u\in C_{v}$. As established above, this also takes $\otilde(\sqrt{n})$ time. 
	
	It is now left to bound the time the procedure spends before finding a cheap $C_{v}$-good color for $v$. We analyze this bound conditioned on the success of Lemma \ref{lemma:number-of-trials-until-good}, i.e., we assume that the procedure performs $r=O(\Delta\log n/s)=O(\Delta/\sqrt{n})$ color trials before finding a cheap $C_{v}$-good color for $v$ (recall that this occurs whp). Consider the sequence $i_{1},\dots, i_{r}$ of colors chosen uniformly at random by the procedure and let $X_{j}=|L_{i_{j}}|$ for each $1\leq j\leq r$. Our goal is to bound $X=\sum_{j=1}^{r}X_{j}$.

    Since all sampled colors are cheap, we have $X_{j}< 10\sqrt{n}$ for each $1\leq j\leq r$, which means that the prerequisite for Bernstein's inequality (Inequality \ref{inequality:bernstein-convenient}) is satisfied for $M=10\sqrt{n}$. We now bound: (i) $\mathbb{E}[X]$; and (ii) $\sum_{j=1}^{r} \textup{Var}(X_{j})$. For (i), first notice that $\mathbb{E}[X_{j}]\leq n/(\Delta+1)$ for every $1\leq j\leq r$. This is because the average list size over all colors is $n/(\Delta+1)$, and this is an upper bound on the average list size over the cheap colors. So, by linearity of expectation, we get $\mathbb{E}[X]=\sum_{j=1}^{r}\mathbb{E}[X_{j}]\leq r\cdot n/(\Delta+1)=O((\Delta/\sqrt{n})\cdot n/\Delta)=O(\sqrt{n})$.
    
    As for (ii), we have $\sum_{j=1}^{r} \textup{Var}(X_{j})=\sum_{j=1}^{r}\mathbb{E}[X_{j}^{2}]-(\mathbb{E}[X_{j}])^{2}\leq \sum_{j=1}^{r} \mathbb{E}[X_{j}^{2}]$. Moreover, since $0\leq X_{j}\leq M$ for each cheap color $1\leq j\leq r$, it follows that
    \begin{align*}
    	\sum_{j=1}^{r} \textup{Var}(X_{j})&\leq \sum_{j=1}^{r} \mathbb{E}[X_{j}^{2}]
    	\leq \sum_{j=1}^{r} \mathbb{E}[M\cdot X_{j}]
    	\\&=M\cdot\sum_{j=1}^{r} \mathbb{E}[X_{j}]=
    	O(\sqrt{n})\cdot \mathbb{E}[X]=O(n).
    \end{align*}

	
	
	For any $t>0$, it follows from Bernstein's inequality that 
	\[\Pr[X>t+\mathbb{E}[X]]\leq \exp\left(-\frac{t^{2}/2}{\sum_{j=1}^{r}\textup{Var}(X_{j})+Mt/3}\right)\leq \exp\left(-\frac{t^{2}/2}{O(n)+O(\sqrt{n})t}\right).\]
	For a suitable choice of $t=\Theta(\sqrt{n}\log n)$, we get the desired bound of
	\[\Pr[X>2t]<\Pr[X>t+\mathbb{E}[X]]\leq \exp\left(-\Omega\left(\frac{(\sqrt{n}\log n)^{2}}{n+\sqrt{n}\cdot \sqrt{n}\log n}\right)\right)=\exp(-\Omega(\log n))=n^{-c}.\] 
	Overall, we get that $X\leq 2t=O(\sqrt{n}\log n)$ whp which concludes the runtime analysis. 
\end{proof}

\subsubsection{Structural Procedures}\label{section:structural}
We present the structural procedures that are implemented during the phase and help maintain the three Invariants listed in Section \ref{section:recoloring}. We start by describing a simple greedy rule for orienting inserted edges between two chosen nodes. Then, we describe scenarios in which the algorithm adds an unchosen node into the chosen set. 

\paragraph{Orienting inserted edges between chosen nodes.} When an edge is inserted between two chosen nodes as a dynamic update, it is oriented towards the endpoint with higher out-degree, breaking ties arbitrarily. If an inserted edge $(u,v)$ is oriented from $u$ to $v$, then $v$ is added to the set $N^{+}(u)$ and $H_{v}$ is updated by increasing the key of $u$'s color by $1$. 

\paragraph{Adding nodes to the chosen set.} Throughout the phase, the algorithm maintains a numerical variable $\mathtt{threshold}$ initially set to $13s$. An unchosen node $v\in V-C$ is added into the set $C$ in one of the following cases: (I) it has more than $\mathtt{threshold}$ chosen neighbors; or (II) it violates Invariant (1), i.e., its degree is larger than $\Delta-s$ and it has fewer than $2s$ chosen neighbors. Whenever a node is added into $C$, $\mathtt{threshold}$ is increased by $1$. If both endpoints of an updated edge $(u,v)$ are added into $C$, it is convenient to assume that they are added one after the other in some arbitrary order (and $\mathtt{threshold}$ is incremented on each addition). Notice that such an event can occur (only) if both endpoints of an inserted edge are unchosen before the insertion, and the insertion led both of them to violate Invariant (1). 

When a node $v$ is added into $C$, its edges to chosen neighbors are oriented outwards. A minimum heap $H_{v}$ is constructed to store the color histogram of $N^{*}(v)$ (which is simply the set of $v$'s unchosen neighbors at this point) and $v$ is added to $C_{u}$ for each unchosen neighbor $u$ of $v$. Finally, if $v$ is colored $i$ at this point, then it is deleted from $L_{i}$.

\paragraph{Analysis.} We now analyze the structural procedures. 
\begin{lemma}\label{lemma:number-of-choosings}
	The number of nodes that are added into $C$ during the phase is at most $2n/s$.
\end{lemma}
\begin{proof}
	We say that a node $v$ is \emph{erratic} if at least $s$ dynamic updates occurred on edges incident on $v$ during the phase. Since each update is incident on $2$ nodes and there are $n$ updates during the phase, it follows that there are at most $2n/s$ erratic nodes in a phase. We complete the proof by showing that every node added to $C$ must be erratic.
	
	First, let us consider a node added into $C$ due to Case (I). Notice that by Property (1) of Lemma \ref{lemma:preprocessing}, in the beginning of the phase each unchosen node has fewer than $12s$ chosen neighbors. Moreover, whenever a node is added to $C$, $\mathtt{threshold}$ is incremented. This means that if a node $v$ surpasses $\mathtt{threshold}$ chosen neighbors during the phase, more than $s$ edges between $v$ and chosen nodes were inserted during the phase. By definition, this means that $v$ is erratic.
	
	Consider now a node $v$ added into $C$ due to Case (II). If $v$ has degree at most $\Delta-2s$ in the beginning of the phase, then to satisfy the conditions of Case (II) it must have its degree increase by more than $s$, making it erratic. Suppose now that $v$'s degree in the beginning of the phase is larger than $\Delta-2s$. By Property (2) of Lemma \ref{lemma:preprocessing}, this means that $v$ has more than $3s$ chosen neighbors in the beginning of the phase. Thus, to satisfy the conditions of Case (II) it must have more than $s$ edges to chosen neighbors deleted, making it erratic.
\end{proof}

In the following two lemmas we prove the maintenance of Invariants (1) and (2) as well as bound the total time associated with adding nodes to $C$. 
\begin{lemma}\label{lemma:invariants-1+2}
	Invariants (1) and (2) are satisfied throughout the phase.
\end{lemma}
\begin{proof}
	Invariant (1) clearly holds since whenever a node violates the invariant, it is added into $C$ (and thus stops violating the invariant). For Invariant (2), it follows from Lemma \ref{lemma:number-of-choosings} that during the phase $\mathtt{threshold}$ increases by at most $2n/s=2\sqrt{n}/\log n< s$. Since $\mathtt{threshold}$ is initialized to $13s$, we get that throughout the phase $\mathtt{threshold}<14s$. The invariant now follows because by construction, each unchosen node has at most $\mathtt{threshold}$ chosen neighbors. 
\end{proof}

\begin{lemma}\label{lemma:runtime-of-choosings}
	The total runtime associated with adding nodes into $C$ during the phase is $O(n\Delta /s)=O(n^{3/2})$.
\end{lemma}
\begin{proof}
	Each addition into $C$ requires building one heap in $O(\Delta)$ time. By Lemma \ref{lemma:number-of-choosings} there are $O(n/s)$ additions, for a total of $O(n\Delta/s)$.
\end{proof}

We complete the analysis by establishing the maintenance of Invariant (3).
\begin{lemma}\label{lemma:invariant-3}
	Invariant (3) is satisfied throughout the phase.
\end{lemma}
\begin{proof}
	We start by observing that if during the phase a chosen node $v$ has out-degree at least $15s$, then it must mean that at least $s$ outgoing edges of $v$ are edges that were inserted during the phase. Indeed, if $v$ is chosen in the beginning of the phase, then by Property (1) of Lemma \ref{lemma:preprocessing} its (out-)degree in $G[C]$ is smaller than $12s$; and if $v$ is added to $C$ during the phase, then upon addition into $C$, its out-degree in $G[C]$ is at most $\mathtt{threshold}$ which is smaller than $14s$. Notice that by construction, once $v$ is in $C$, every new outgoing edge of $v$ must be an edge that was inserted during the phase.  
	
	Now, assume by contradiction that at some point during the phase, there exists a node $v\in C$ with out-degree larger than $16s$. By the way the orientation is constructed, this must mean that $v$ has at least $s$ outgoing edges to nodes whose out-degree at the time of insertion is at least $15s$. By the observation above, each of these nodes is incident on at least $s$ outgoing edges that were inserted during the phase. In particular, this means that there have been at least $s^{2}>n$ edge insertions during the phase which contradicts the phase length. 
\end{proof}

	

\section{Algorithm with $\otilde(\sqrt{n})$ Worst-Case Update Time}\label{section:worst-case-alg}
In this section, we present an algorithm that achieves a worst-case update time of $\otilde(\sqrt{n})$ as stated in Theorem \ref{theorem:worst-case}. The algorithm is constructed by deamortizing the algorithm presented in Section \ref{section:amortized-alg}. To that end, we identify the sources of amortization in the algorithm and describe the techniques used to deamortize them. Namely, we show how to deamortize: (1) the additions of nodes into the chosen set $C$ during the phase; and (2) the preprocessing stage. 

The deamortization of (1) is presented in Section \ref{section:deamortizing-choosings}. The idea is to first classify the nodes added into $C$ during a phase of the algorithm of Section \ref{section:amortized-alg}. These are the \emph{overloaded} nodes, i.e., nodes with ``too many'' chosen neighbors; and the \emph{underloaded} nodes, i.e., nodes with ``too few'' chosen neighbors. As we show, by slightly modifying our recoloring procedures, we can overcome the problems posed by each of these node types without adding them into the chosen set $C$. The modified recoloring procedures maintain the same asymptotic runtime guarantees as the original procedures, leading to the desired deamortization result. 

The deamortization of (2) is presented in Section \ref{section:deamortizing-sampling-chosen}. The idea is that we can pipeline the preprocessing of phase $i+1$ across phase $i$. That is, instead of a static preprocessing algorithm in the beginning of each phase, we show how its execution can be distributed along the previous phase while maintaining the same guarantees as the static preprocessing. The analysis of the resulting algorithm is given in Section \ref{section:worst-case-analysis}. Finally, in Section \ref{section:knowledge} we describe the adjustments required to implement the algorithm without the foreknowledge of $\Delta$.

\subsection{Deamortizing Additions Into the Chosen Set}\label{section:deamortizing-choosings}

Recall that the addition of a node $v$ into $C$ can occur in one of two cases (see Section \ref{section:structural}): (I) $v$ has more than $\mathtt{threshold}$ chosen neighbors; or (II) $v$'s degree is larger than $\Delta-s$ and it has fewer than $2s$ chosen neighbors. To get a worst-case bound, we cannot afford to have these additions into $C$ as each addition may take $\Theta(\Delta)$ time.

We present a modified algorithm that does not add unchosen nodes into $C$ during the phase. Towards that goal, we define the following notions. We say that an unchosen node is \emph{overloaded} if at some point during the phase it has more than $13s$ chosen neighbors.\footnote{Notice that here, the threshold remains fixed. In contrast, in the definition of Section \ref{section:structural} the threshold increases after each addition into $C$. The reason is that in Section \ref{section:structural}, if the threshold remained static it could have caused a cascading effect leading to many nodes to surpass the threshold (due to the addition of their neighbors into $C$). This is not a danger in the current section since we do not add nodes into $C$.} An unchosen node is referred to as \emph{underloaded} if at some point during the phase it satisfies Case (II). Throughout the phase, the algorithm maintains a list $\ell_{\textup{ovr}}$ of the overloaded nodes, and a balanced tree $T_{\text{und}}$ of the underloaded nodes. It is not hard to see that maintaining $\ell_{\textup{ovr}}$ and $T_{\textup{und}}$ can be done well within the desired update time.

It is important to note that while the algorithm of Section \ref{section:amortized-alg} handles overloaded and underloaded nodes in the same manner (by adding them into $C$), they are problematic for very different reasons. In the case of an overloaded node, the problem is that changing its color requires updating the heaps $H_{u}$ of all chosen neighbors which becomes too costly; whereas in the case of an underloaded node, the problem is that finding a cheap $C_{v}$-good color may take too long as it may not satisfy the conditions of Observation \ref{observation:good-colors}. We now explain how to handle each type of nodes.

\subsubsection{Handling Overloaded Nodes}\label{section:overloaded}
 Once a node $u\in V$ becomes overloaded, the algorithm stops taking its color into account in the heaps $H_{v}$ of its chosen neighbors $v\in C_{u}$. In other words, for each chosen node $v\in C$, the algorithm maintains in $H_{v}$ the color histogram only of the nodes in $N^{*}(v)$ that are not overloaded.    

Recall that when a chosen node undergoes a recoloring, it collects the color histogram of $N(v)$ into $H_{v}$. In the algorithm of Section \ref{section:amortized-alg}, this was accomplished simply by going over the nodes in $N^{+}(v)$. In the current modified algorithm, $v$ has to also take into account its overloaded neighbors. This is done by simply going over the entire list $\ell_{\textup{ovr}}$ and adding the colors of $v$'s overloaded neighbors into the color histogram of $N(v)$.

One final detail that we need to take care of is in the implementation of a recoloring for an overloaded node $v$. In particular, when $v$ gets recolored, the algorithm looks for a $C_{v}$-good color. In the standard recoloring procedure of Section \ref{section:recoloring}, this is implemented by first collecting the color histogram of $C_{v}$. In our case, collecting this color histogram from scratch becomes too costly, and so we avoid this step as follows. At the moment $v$ becomes overloaded, the algorithm creates a designated data structure for $v$ and collects the color histogram of $C_{v}$ into it. Notice that this only takes $\otilde(s)$ time and thus, does not violate the runtime guarantee. Afterwards, whenever $v$ has a chosen neighbor added (resp., removed) its color is added to (resp., removed from) the color histogram of $C_{v}$. In addition, when a chosen node $v\in C$ is recolored, the algorithm updates the color change in the color histograms of all of $v$'s overloaded neighbors (which are found by going over $\ell_{\textup{ovr}}$).

The following observation is straightforward from the construction.
\begin{observation}\label{observation:overloaded}
	The runtime of the modified procedure for recoloring chosen nodes is $\otilde(\sqrt{n}+|\ell_{\textup{ovr}}|)$. 
\end{observation}

\subsubsection{Handling Underloaded Nodes}\label{section:underloaded}
Since the underloaded nodes might be hard to recolor based on the existing procedure for recoloring unchosen nodes, we describe a modified recoloring procedure specifically for underloaded nodes. That is, the modified algorithm has two separate procedures for unchosen nodes --- one that colors unchosen nodes that are not underloaded (which is the one described in Section \ref{section:recoloring}) and another one for underloaded nodes. The idea of the procedure is to have the underloaded node steal a color from a neighbor which is not underloaded and then recolor that neighbor (which is easier to recolor since it is not underloaded). 

More concretely, suppose that $v$ is an underloaded node which is currently colored $i$ and is being recolored. Let $U_{v}$ be the set of $v$'s underloaded neighbors and let $W_{v}=N(v)-U_{v}$. The algorithm samples colors until finding a cheap $W_{v}$-good color for $v$. To check if a cheap color $j$ is $W_{v}$-good, the algorithm goes over $L_{j}$. If $L_{j}$ contains more than one neighbor of $v$ or an underloaded neighbor $u\in U_{v}$, then $j$ is not a $W_{v}$-good color. Otherwise, the algorithm adds the color histogram of $C_{v}$ and checks that in total, there is at most one node colored with $j$ in $W_{v}$ (which means that $j$ is a $W_{v}$-good color). Notice that for the latter step, the algorithm only needs to construct the color histogram of $C_{v}$ once before the random color trials. Moreover, it could be that an underloaded node later becomes overloaded in which case the color histogram of $C_{v}$ is already prepared (as discussed above).

When a cheap $W_{v}$-good color $j$ is found, $v$ gets colored with that color and the algorithm modifies the relevant data structures similarly to the procedure of Section \ref{section:recoloring}. If there is a node $w\in N(v)$ colored $j$, then it is recolored either by the procedure for recoloring unchosen nodes defined in Section \ref{section:recoloring} (in the case that $w\notin C$); or by the procedure for recoloring chosen node defined in Section \ref{section:overloaded} (in the case that $w\in C$). Notice that by definition, there can only be one such node $w$ and it is in $W_{v}$, i.e., it is not underloaded.

This concludes the description of the procedure for recoloring underloaded nodes. We establish the following simple observation. 
\begin{observation}\label{observation:underloaded}
	Let $v$ be an underloaded node and suppose that $|U_{v}|\leq \Delta/5$. Then, there are at least $\Delta/10$ cheap colors that are $W_{v}$-good for $v$.
\end{observation} 
\begin{proof}
	First, if $\deg(v)<4\Delta/5$, then at least $\Delta/5$ colors are available for $v$ and thus, also $W_{v}$-good for $v$. Since the total number of expensive colors is at most $\sqrt{n}/10<\Delta/10$, we get that at least $\Delta/5-\Delta/10=\Delta/10$ cheap colors are $W_{v}$-good for $v$. Now, suppose that $\deg(v)\geq 4\Delta/5$. Then, $|W_{v}|=\deg(v)-|U_{v}|\geq 3\Delta/5$. By Observation \ref{observation:q-good}, more than $3\Delta/10$ colors are $W_{v}$-good for $v$. Since there are fewer than $\Delta/10$ expensive colors in total, we get that there are at least $\Delta/5>\Delta/10$ cheap colors that are $W_{v}$-good for $v$.
\end{proof}

The following lemma now follows directly from Observation \ref{observation:underloaded}.
\begin{lemma}\label{lemma:unedrloaded}
	Suppose that $v$ is an underloaded node that undergoes a recoloring procedure and suppose that $|U_{v}|\leq \Delta/5$. Then, within $O(\log n)$ random color trials, a cheap $W_{v}$-good color for $v$ is found whp.
\end{lemma} 

\subsection{Deamortizing the Preprocessing Stage}\label{section:deamortizing-sampling-chosen}
The second component that we seek to deamortize is the preprocessing of each phase. The idea here is that instead of a preprocessing stage in the beginning of each phase, the preprocessing of phase $i+1$ is implemented during the execution of phase $i$.\footnote{A small issue that arises here is with the preprocessing of the first phase. To address this issue, note that since the first phase starts from an empty graph, the preprocessing can be implemented simply by setting all nodes as chosen in that phase.} This implementation occurs ``behind the scenes'', in parallel to the execution of the current phase. We refer to this module as the \emph{behind the scenes (bts)} algorithm, whereas the execution of the other modules is referred to as the main algorithm. We emphasize that we \emph{do not} assume that the decisions made as part of the bts algorithm are hidden from the adaptive adversary.  

Consider phase $t$ of the main algorithm's execution and let $C^{t}$ denote the set of chosen nodes during this phase. We now describe how the set $C^{t+1}$ and the associated data structures of phase $t+1$ are computed by the bts algorithm during phase $t$.\footnote{At a first glance, the algorithm may seem to be space-inefficient as it defines designated data structures per phase. However, we note that during phase $t$, there is no need to store the data structures associated with phases $t'<t$ in the memory.} In the beginning of phase $t$, the set $C^{t+1}$ is empty. Let $V_{1},\dots, V_{\sqrt{n}}$ be an arbitrary partition of $V$ into $\sqrt{n}$ disjoint batches, each consisting of at most $\sqrt{n}$ nodes. The bts algorithm uses the first $\sqrt{n}$ updates of the phase to sample nodes into $C^{t+1}$. In particular, for each $1\leq i\leq \sqrt{n}$, during the $i$-th update the bts algorithm samples each node $v\in V_{i}$ into the set $C^{t+1}$ independently with probability $10s/\Delta$. If $v$ is not sampled into $C^{t+1}$ and its current color is $j$, then it is inserted to the list $L^{t+1}_{j}$ of unchosen nodes (w.r.t.\ $C^{t+1}$) colored $j$.

Following the sampling of $C^{t+1}$, the bts algorithm goes on to construct an orientation of $G[C^{t+1}]$ and a minimum heap $H^{t+1}_{v}$ for each $v\in C^{t+1}$. To that end, the execution is divided into \emph{blocks} of $O(\Delta/s)$ updates, where each block is used to construct the minimum heap $H^{t+1}_{v}$ of a distinct node $v\in C^{t+1}$. Let $v_{1},\dots, v_{|C^{t+1}|}$ be an arbitrary ordering of the nodes in $C^{t+1}$ (e.g., the order in which they were sampled into $C^{t+1}$). During the $i$-th block, the bts algorithm goes over the neighbors of node $v=v_{i}$ in batches of $O(s)$ nodes. For each unchosen neighbor $u$ of $v$
that is not overloaded with respect to $C^{t+1}$, 
the algorithm adds $u$'s color into $H^{t+1}_{v}$. If a neighbor $u$ of $v$ belongs to $C^{t+1}$ and appears before $v$ in the ordering, then the edge is oriented based on the greedy rule described in Section \ref{section:structural} (otherwise, the edge orientation is executed during $u$'s block). Whenever an edge is oriented from a node $u$ to a node $v$, the color of $u$ is added to $H^{t+1}_{v}$.  

Additional data structures maintained by the bts algorithm are the list $\ell^{t+1}_{\textup{ovr}}$ associated with nodes that are overloaded w.r.t.\ $C^{t+1}$ and the balanced tree $T^{t+1}_{\textup{und}}$ associated with nodes that are underloaded w.r.t.\ $C^{t+1}$. Notice that the bts algorithm is required to maintain the data structures associated with phase $t+1$ correctly during phase $t$. To achieve that, the bts algorithm updates these data structures in response to insertion/deletions during phase $t$. This is done in an analogous way to how the main algorithm updates its data structures.

The following Lemma captures the main properties achieved by the bts algorithm.
\begin{lemma}\label{lemma:bts}
	Consider some phase $t>1$. From the point the bts algorithm finishes sampling the set $C^{t}$ (during phase $t-1$) until the end of phase $t$, the following properties hold whp: (i) The out-degree of each $v\in C^{t}$ in $G[C^{t}]$ is $O(s)$; (ii) at most $\sqrt{n}$ nodes are overloaded w.r.t.\ $C^{t}$; and (iii) at most $\sqrt{n}$ nodes are underloaded w.r.t.\ $C^{t}$.
\end{lemma}
\begin{proof}
	Notice that the sampling of $C^{t}$ implemented in the first $\sqrt{n}$ updates of phase $t-1$ is similar to the global sampling of Section \ref{section:preprocessing}, apart from the fact that the adversary may decide to insert/delete edges after seeing (parts of) the sample. However, since the entire sampling occurs within $\sqrt{n}$ updates, the adversary can only insert/delete $\sqrt{n}=o(s)$ edges. In particular, when the bts algorithm finishes sampling $C^{t}$, it satisfies Properties (1) and (2) of Lemma \ref{lemma:preprocessing} whp. After that point, there are only $n'= 2n-\sqrt{n}$ updates until the end of phase $t$. Recall that Lemma \ref{lemma:number-of-choosings} shows that in $n$ updates, at most $2n/s$ nodes become underloaded/overloaded. So, plugging in $n'$, we get that a bound of $4n/s<4\sqrt{n}/\log n <\sqrt{n}$ on the size of both sets, implying Properties (ii) and (iii). Similarly, it follows from Lemma \ref{lemma:invariant-3} that after $n'$ updates, the out-degree of each node remains $O(s+\sqrt{n'})=O(s)$ implying Property (i).
\end{proof}

\subsection{Analysis}\label{section:worst-case-analysis}
We provide the proof of Theorem \ref{theorem:worst-case}.
\begin{proof}[Proof of Theorem \ref{theorem:worst-case}]
	It is straightforward from the construction that each recoloring procedure leads to at most $3$ nodes being recolored (which can be exactly $3$ when an underloaded node is recolored). Hence, the recourse is $3$. Furthermore, the correctness of the modified recoloring procedures follows by similar arguments to the ones used to show the correctness of the recoloring procedures of Section \ref{section:recoloring}. So, it is left to establish the update time.
	
	We first consider the bts algorithm of Section \ref{section:deamortizing-sampling-chosen} during some phase $t$. Notice that when constructing the set $C^{t+1}$, the bts algorithm goes over $\sqrt{n}$ nodes per update and performs an independent sampling. Similarly, to construct the edge orientation of $G[C^{t+1}]$ and the data structures $H_{v}^{t+1}$, the bts algorithm performs $O(\Delta)$ operations in blocks of $O(\Delta/s)$ updates such that in each update, $O(s)$ operations are performed. Regarding the changes the bts algorithm makes in response to edge insertions/deletions, notice that these are done in the same way as the main algorithm. Since these have a worst-case update time of $\otilde(\sqrt{n})$ in the main algorithm, the same holds for the bts algorithm.
	
	Regarding the runtime of the modified recoloring procedure for chosen nodes (Section \ref{section:overloaded}), we first recall that by Property (ii) of Lemma \ref{lemma:bts}, the number of overloaded nodes is bounded by $\sqrt{n}$ whp. Combined with Observation \ref{observation:overloaded}, we get that the procedure's runtime is $\otilde(\sqrt{n})$.
	
	We complete the proof by analyzing the recoloring procedure for underloaded nodes (Section \ref{section:underloaded}). First, note that by Property (iii) of Lemma \ref{lemma:bts}, the total number of underloaded nodes during a phase is at most $\sqrt{n}$ whp. In particular, this means that every node $v\in V$ has $|U_{v}|\leq \sqrt{n}< \Delta/5$ underloaded neighbors. Hence, by Lemma \ref{lemma:unedrloaded}, a cheap $W_{v}$-good color is found within $O(\log n)$ color trials. Since each sampled color $i$ satisfies $|L_{i}|=O(\sqrt{n})$, it follows that the runtime until finding a cheap $W_{v}$-good color is $\otilde (\sqrt{n})$. After finding a cheap $W_{v}$-good color, the algorithm may recolor a single non-underloaded node (which could be chosen or unchosen). As discussed in Lemma \ref{lemma:recoloring-time}, this recoloring takes $\otilde (\sqrt{n})$ time. 
\end{proof}

\subsection{Lifting the Knowledge of $\Delta$ Assumption}\label{section:knowledge}

The algorithm of Theorem \ref{theorem:worst-case} (as well as the algorithm of Theorem \ref{theorem:amortized-alg-statement}) has been presented under the simplifying assumption that the maximum degree $\Delta$ is known in advance. This assumption was used in the algorithm for the sampling probability of chosen nodes (which was set to $10s/\Delta$), the length of blocks in which these nodes are processed in the bts algorithm (which was set to $\Delta/s$), and the length of heaps $H_{v}$ of chosen nodes (which was set to $\Delta+1$). Additionally, the assumption allows the algorithm to distinguish in advance between the case of $\Delta<20\sqrt{n}\log n$ (in which the trivial $O(\Delta)$ algorithm is employed), and the case of $\Delta\geq 20\sqrt{n}\log n$ (in which the new algorithm is employed). In this section, we present the adjustments that allow the algorithm to work even without knowing $\Delta$ in advance.

Throughout this subsection, every histogram \(H_v\) is implemented as a balanced search tree storing only colors with positive counts, with the augmentation described below, in place of the minimum heaps used earlier.

Throughout the execution, the algorithm maintains the maximum degree observed so far as a variable $\hat{\Delta}$. At all times, the algorithm uses at most $\hat{\Delta}+1$ colors. We naturally describe the algorithm as the following two consecutive components: the first component which is implemented as long as  $\hat{\Delta}<20\sqrt{n}\log n$, and the second component which is implemented from the point $\hat{\Delta}$ reaches $20\sqrt{n}\log n$ and onward. 

During the first component, whenever the algorithm needs to recolor a node, it does so in $O(\hat{\Delta})$ time based on the trivial algorithm's update rule. In addition, the first component is used for the preprocessing of the first phase in the second component. To that end, throughout the first component the algorithm maintains an (arbitrary) orientation of the graph $G$, and for each node $v\in V$, stores the color histogram of $N^{*}(v)=N(v)-N^{+}(v)$ in a 
balanced search tree $H_v$ storing only colors with positive counts, augmented to find an absent color in logarithmic time as described below. We use this representation throughout both components.
Notice that because the maximum degree is bounded by $20\sqrt{n}\log n$, the update time in the first component is $\otilde(\sqrt{n})$.

Once $\hat{\Delta}$ reaches $20\sqrt{n}\log n$, the algorithm switches to the second component in which the algorithm of Theorem \ref{theorem:worst-case} is implemented. In the first phase of the second component, the algorithm sets the set of chosen nodes as $C^{1}=V$ and uses the edge orientation and data structures defined from the first component. In every phase $t$ of the second component, the algorithm implements the preprocessing for phase $t+1$ based on the following adjustments to the bts algorithm. First, the chosen nodes $C^{t+1}$ of phase $t+1$ are sampled with probability $10s/\hat{\Delta}^{t}$, where $\hat{\Delta}^{t}$ is the value of $\hat{\Delta}$ measured in the beginning of phase $t$. Following that, let $v_{1},\dots, v_{|C^{t+1}|}$ be the ordering in which the nodes of $C^{t+1}$ are processed. The block size of node $v_{i}$ is adjusted to $\lceil \deg(v_{i})/s\rceil$ where $\deg(v_{i})$ is measured at the starting time of $v_{i}$'s processing. The $H^{t+1}_{v}$ data structures are implemented as balanced trees that only explicitly maintain the colors that have positive number of occurrences in the color histogram. By having each tree-vertex keep track of the size, minimum color, and maximum color of its subtree, a node can find an available color in logarithmic time.

It may seem at first glance that since the adversary can arbitrarily increase the maximum degree, the sampling probability which was set based on the previous maximum degree can be too high and cause a violation to some of the algorithm's properties. To see that this is not the case, notice that the issue with adversarial increase to the degree is that some nodes may have ``too many" chosen neighbors. Note however, that this type of nodes was already accounted for in the algorithm that assumes knowledge of $\Delta$ --- these were the overloaded nodes. Moreover, Lemma \ref{lemma:bts} suggests that there are at most $\sqrt{n}$ such nodes in a phase whp. Notice that this statement holds regardless of the increase to the maximum degree. Therefore, the algorithm's properties are preserved and we can use a similar analysis to that of Theorem \ref{theorem:worst-case} to derive the update time bound and correctness.

	
	\subsection*{AI Statement}
	The authors developed all the results and wrote the paper without AI assistance. In developing this full version, Opus 5.5 was used to detect minor gaps in arguments and presentation.
	\clearpage
	\bibliographystyle{alpha}
	
	\bibliography{references}
\end{document}